\documentclass[copyright,creativecommons]{eptcs}
\providecommand{\event}{EXPRESS 2017} 
\usepackage{breakurl}
\usepackage{xspace}
\usepackage{graphicx}
\usepackage{amsmath,amsfonts,amssymb}
\usepackage{wrapfig} 
\usepackage{enumitem}
\setitemize{noitemsep,topsep=0pt,parsep=0pt,partopsep=0pt,leftmargin=0.7cm}
\setenumerate{noitemsep,topsep=0pt,parsep=0pt,partopsep=0pt,leftmargittn=0.7cm}
\usepackage{multicol}
\usepackage{newtxmath}
\usepackage[T1]{fontenc}		              		
\usepackage{xcolor}

\DeclareSymbolFont{frenchscript}{OMS}{ztmcm}{m}{n}
\DeclareMathSymbol{\A}{\mathord}{frenchscript}{65}   
\DeclareMathSymbol{\HC}{\mathord}{frenchscript}{72}  
\DeclareMathSymbol{\K}{\mathord}{frenchscript}{75}   
\DeclareMathSymbol{\Sig}{\mathord}{frenchscript}{83} 
\DeclareMathSymbol{\Exp}{\mathord}{frenchscript}{84} 
\newcommand{\T}[1]{\Exp_{{\rm #1}}}              
\newcommand{\RL}{L}                            

\DeclareMathOperator*{\ParOp}{\mid}
\DeclareMathOperator*{\BigParOp}{\Big|}

\newcommand{\ncA}{{\bf noncritA}}
\newcommand{\crA}{{\bf critA}}
\newcommand{\ncB}{{\bf noncritB}}
\newcommand{\crB}{{\bf critB}}

\makeatletter
\def\moverlay{\mathpalette\mov@rlay}
\def\mov@rlay#1#2{\leavevmode\vtop{%
   \baselineskip\z@skip \lineskiplimit-\maxdimen
   \ialign{\hfil$\m@th#1##$\hfil\cr#2\crcr}}}
\newcommand{\charfusion}[3][\mathord]{
    #1{\ifx#1\mathop\vphantom{#2}\fi
        \mathpalette\mov@rlay{#2\cr#3}
      }
    \ifx#1\mathop\expandafter\displaylimits\fi}
\makeatother
\newcommand{\dcup}{\charfusion[\mathbin]{\cup}{\mbox{\Large$\cdot$}}}

\newcommand{\nil}{\ensuremath{\textbf{0}}}
\newcommand{\signals}{\ensuremath{\mathrel{\hat{}\!}}}

\newcommand{\CCSs}{CCSS\xspace} 
\makeatletter
\def\comesfrom{\@transition\leftarrowfill}
\def\goesto{\@transition\rightarrowfill}
\def\ngoesto{\@transition\nrightarrowfill}
\def\Goesto{\@transition\Rightarrowfill}
\def\nGoesto{\@transition\nRightarrowfill}
\def\xmapsto{\@transition\mapstofill}
\def\nxmapsto{\@transition\nmapstofill}
\def\@transition#1{\@@transition{#1}}
\newbox\@transbox
\newbox\@arrowbox
\newbox\@downbox
\def\@@transition#1#2%
   {\setbox\@transbox\hbox
      {\vrule height 1.5ex depth .9ex width 0ex\hskip0.25em$\scriptstyle#2$\hskip0.25em}
   \ifdim\wd\@transbox<1.5em
      \setbox\@transbox\hbox to 1.5em{\hfil\box\@transbox\hfil}\fi
   \setbox\@arrowbox\hbox to \wd\@transbox{#1}
   \ht\@arrowbox\z@\dp\@arrowbox\z@
   \setbox\@transbox\hbox{$\mathop{\box\@arrowbox}\limits^{\box\@transbox}$}
   \dp\@transbox\z@\ht\@transbox 10pt
   \mathrel{\box\@transbox}}
\def\nrightarrowfill{$\m@th\mathord-\mkern-6mu%
  \cleaders\hbox{$\mkern-2mu\mathord-\mkern-2mu$}\hfill
  \mkern-6mu\mathord\not\mkern-2mu\mathord\rightarrow$}
\def\Rightarrowfill{$\m@th\mathord=\mkern-6mu%
  \cleaders\hbox{$\mkern-2mu\mathord=\mkern-2mu$}\hfill
  \mkern-6mu\mathord\Rightarrow$}
\def\nRightarrowfill{$\m@th\mathord=\mkern-6mu%
  \cleaders\hbox{$\mkern-2mu\mathord=\mkern-2mu$}\hfill
  \mkern-6mu\mathord\not\mathord\Rightarrow$}
\def\mapstofill{$\m@th\mathord\mapstochar\mathord-\mkern-6mu%
  \cleaders\hbox{$\mkern-2mu\mathord-\mkern-2mu$}\hfill
  \mkern-6mu\mathord\rightarrow$}
\def\nmapstofill{$\m@th\mathord\mapstochar\mathord-\mkern-6mu%
  \cleaders\hbox{$\mkern-2mu\mathord-\mkern-2mu$}\hfill
  \mkern-6mu\mathord\not\mkern-2mu\mathord\rightarrow$}
\makeatother 
\newcommand{\ar}[1]{\mathrel{\goesto{#1}}}            

\def\squareforqed{\hbox{\rlap{$\sqcap$}$\sqcup$}}
\def\qed{\ifmmode\squareforqed\else{\unskip\nobreak\hfil
\penalty50\hskip1em\null\nobreak\hfil\squareforqed
\parfillskip=0pt\finalhyphendemerits=0\endgraf}\fi}
\newtheorem{theorem}{Theorem}
\newtheorem{definition}{Definition}

\newenvironment{proof}{\noindent{\it Proof.}\ }{\qed}

\newcommand{\plat}[1]{\raisebox{0pt}[0pt][0pt]{#1}}  

\newcommand{\ass}[2]{{\it asgn}_{#1}^{\,#2}}
\newcommand{\noti}[2]{{\it n}_{#1}^{\,#2}}

\newcommand{\rA}{{\it readyA}}
\newcommand{\rB}{{\it readyB}}
\newcommand{\tu}{{\it turn}}
\newcommand{\tr}{{\it true}}
\newcommand{\fa}{{\it false}}

\newcommand{\Tu}[1]{{\it Turn}^{#1}}
\newcommand{\RA}[1]{{\it ReadyA}^{\,#1}}
\newcommand{\RB}[1]{{\it ReadyB}^{\,#1}}

\newcommand{\procA}{{\rm A}\xspace}
\newcommand{\procB}{{\rm B}\xspace}
\newcommand{\procC}{{\rm C}\xspace}

\newcommand{\pilive}{\rho}

\def\titlerunning{Analysing Mutual Exclusion using Process Algebra with Signals}
\title{\mbox{\!\!\!}\titlerunning\!\!\!}
\author{
Victor Dyseryn
  \institute{Ecole Polytechnique, Paris, France}
  \institute{Data61, CSIRO, Sydney, Australia}
  \email{\hspace{-4pt}victor.dyseryn-fostier@polytechnique.edu\hspace{-4pt}}
\and
Rob van Glabbeek \& Peter H\"ofner
  \institute{Data61, CSIRO, Sydney, Australia}
  \institute{School of Computer Science and Engineering\\
  University of New South Wales, Sydney, Australia}%
  \email{\hspace{-4pt}rvg@cs.stanford.edu\qquad Peter.Hoefner@data61.csiro.au\hspace{-4pt}}
}
\def\authorrunning{V. Dyseryn, R.J. van Glabbeek \& P. H\"ofner}

\begin{document}
\maketitle

\begin{abstract}
In contrast to common belief, the Calculus of Communicating Systems (CCS) and similar process algebras lack
the expressive power to accurately capture mutual exclusion protocols
without enriching the language with fairness assumptions. 
Adding a fairness assumption to implement a mutual exclusion protocol seems counter-intuitive. 
We employ a signalling operator, which can be combined with CCS, or other process calculi, 
and show that this minimal extension is expressive enough to model mutual exclusion: we confirm the correctness of Peterson's mutual exclusion algorithm for two
processes, as well as Lamport's bakery algorithm, under reasonable assumptions on the underlying memory model.
The correctness of Peterson's algorithm for more than two processes requires stronger,
less realistic assumptions on the underlying memory model.
\end{abstract}

\section{Introduction}
In the process algebra community it is common belief that, on some level of abstraction, any distributed system can be modelled in standard process-algebraic specification formalisms like the Calculus of Communicating Systems (CCS)~\cite{Mi89}.

However, this sentiment has been proven incorrect~\cite{GH15}: two of the authors presented a simple fair scheduler---one that in suitable variations occurs in many distributed systems---of which \emph{no} implementation can be expressed in CCS, unless CCS is enriched with a fairness assumption.
Instances of our fair scheduler, that hence cannot be rendered correctly, are the \emph{First in First out}%
\footnote{Also known as First Come First Served (FCFS)},
\emph{Round Robin}, and
\emph{Fair Queueing} 
scheduling algorithms\footnote{\url{http://en.wikipedia.org/wiki/Scheduling_(computing)}}
as used in network routers~\cite{rfc970,Nagle87} and operating systems~\cite{Kleinrock64},
or the \emph{Completely Fair Scheduler}\footnote{\url{http://en.wikipedia.org/wiki/Completely_Fair_Scheduler}},
which is the default scheduler of the Linux kernel since version 2.6.23.
Since fair schedulers can be implemented in terms of mutual exclusion,
this result implies that mutual exclusion protocols,
such as the ones by Dekker~\cite{EWD35,EWD123}, Peterson~\cite{Peterson81} and Lamport \cite{bakery}, cannot be rendered correctly in CCS without imposing a fairness assumption.

Close approximations of Dekker's and Peterson's
protocols rendered in CCS or similar formalisms abound in the literature
\cite{Walker89,Bouali91,Valmari96,EsparzaBruns96,AcetoEtAl07}. 
Unless one makes a fairness assumption these renderings do not
possess the liveness property that when a process leaves its
non-critical section, and thus wants to enter the critical section, it
will eventually succeed in doing so. When assuming fairness, this problem disappears \cite{CorradiniEtAl09}.
However, since mutual exclusion protocols are often employed to ensure that each of several
tasks gets allocated a fair amount of a shared resource, assuming fairness to implement
mutual exclusion appears counter-intuitive.  

Informally speaking, the reason why the CCS rendering of
algorithms such as Peterson's does not work, is that it is possible that a process never gets a
chance to write to a shared variable to indicate interest in entering the critical section.\pagebreak[3]
This is because other processes running in parallel and competing for the critical section are `too
busy' reading the shared variable all the time.

In this paper we extend CCS with
\emph{signals}. This extension is able to express
mutual exclusion protocols \emph{without} the use of fairness assumptions.
To prove correctness, one only needs basic assumptions such as progress
and justness.

We will use this extension to analyse the correctness of some of the most famous protocols for
mutual exclusion, namely Peterson's algorithm, the filter lock algorithm---Peterson's algorithm for more
than two processes---and Lamport's bakery algorithm. With regards to the filter lock algorithm our
analysis reveals some surprising protocol behaviour.

\section{Preliminaries: The Calculus of Communicating Systems}\label{sec:ccs}
The Calculus of Communicating Systems (CCS)~\cite{Mi89} is a process
algebra, which is used to describe concurrent processes.

It is parametrised with sets ${\A}$ and ${\K}$ of \emph{names} and \emph{agent identifiers}.
We define the set of  
\emph{handshake actions} as $\HC:=\A \dcup \bar\A$, where $\bar{\A} := \{ \bar{a} \mid a \in \A \}$ 
is the set of \emph{co-names}.
Complementation is extended to $\HC$ by setting $\bar{\bar{a}} = a$.
Finally, \plat{$Act := \HC\dcup \{\tau\}$} is the set of {\em actions}, where $\tau$ is a special \emph{internal action}.
In this paper $a,b,c,\dots$ range over $\HC$, $\alpha,\beta$ over $Act$, and $A$, $B$ range over $\K\!$.
A \emph{relabelling} is a function $f\!:\HC\mathbin\rightarrow \HC$ satisfying
$f(\bar{a})=\overline{f(a)}$; it extends to $Act$ by $f(\tau):=\tau$.
Each $A\in\K$ comes with a defining equation \plat{$A \stackrel{{\it def}}{=} P$}
with $P$ being a CCS expression as defined below.

The class $\T{CCS}$ of \emph{CCS expressions} is defined as the smallest class that includes

\vspace{-2ex}
\begin{multicols}{3}
    \begin{itemize}
        \item \emph{agent identifiers} $A\in\K$\,;
        \item \emph{prefixes} $\alpha.P$\,;
        \item \emph{(infinite) choices} $\sum_{i\in I} P_i$\,;
        \item \emph{parallel compositions} $P|Q$\,;
        \item \emph{restrictions} $P\backslash \RL$\,;
        \item \emph{relabellings} $P[f]$\,;
        \item[]
    \end{itemize}
 \end{multicols}
 
\vspace{-2ex}
 \noindent where $P,P_{i},Q\in\T{CCS}$ are CCS expressions, $I$ an index set,
$\RL\subseteq \A$ a set of names, and $f$ an arbitrary relabelling function.
 In case $I=\{1,2\}$, we write $P_1+P_2$ for $\sum_{i\in I} P_i$. The \emph{inactive process}
 $\nil$ is defined by $\sum_{i\in \emptyset} P_i$; it is not capable to perform any action.

The semantics of CCS is given by the labelled transition relation
$\mathord\rightarrow \subseteq \T{CCS}\times Act \times\T{CCS}$, where transitions 
\plat{$P\ar{\alpha}Q$}
are derived from the rules of \autoref{tab:CCSsos}.
\begin{table}[t]
\normalsize
\begin{center}
\framebox{$\begin{array}{ccc}
\alpha.P \ar{\alpha} P &
&\displaystyle\frac{P_j\ar{\alpha} P'}{\sum_{i\in I}P_i \ar{\alpha} P'}~~(j\mathbin\in I)\\[3.3ex]
\displaystyle\frac{P\ar{\alpha} P'}{P|Q \ar{\alpha} P'|Q} &
\displaystyle\frac{P\ar{a} P' ,~ Q \ar{\bar{a}} Q'}{P|Q \ar{\tau} P'| Q'} &
\displaystyle\frac{Q \ar{\alpha} Q'}{P|Q \ar{\alpha} P|Q'}\\[3.3ex]
\displaystyle\frac{P \ar{\alpha} P'}{P\backslash \RL \ar{\alpha}P'\backslash \RL}~~
                     (\alpha,\bar{\alpha}\not\in\RL) &
\displaystyle\frac{P \ar{\alpha} P'}{P[f] \ar{f(\alpha)} P'[f]} &
\displaystyle\frac{P \ar{\alpha} P'}{A\ar{\alpha}P'}~~(A \stackrel{{\it def}}{=} P)
\end{array}$}
\end{center}
\vspace{-3ex}
\caption{Structural operational semantics of CCS}
\label{tab:CCSsos}
\end{table}
\noindent 
The process $\alpha.P$ performs the action $\alpha$ first and subsequently acts as $P$.
The choice operator $\sum_{i\in I}P_{i}$ may act as any of the $P_i$, depending on which of the processes is able to act at all.
The parallel composition $P|Q$ executes an action from $P$, an action from $Q$, or in the case where
$P$ and $Q$ can perform complementary actions $a$ and $\bar{a}$, the process can perform a synchronisation, resulting in an internal action $\tau$.
The restriction operator $P \backslash \RL$
inhibits execution of the actions from $\RL$ and their complements. 
The only way for a subprocess of $P \backslash \RL$ to perform an action $a\in \RL$ is through
synchronisation with another subprocess of $P \backslash \RL$, which performs $\bar a$. 
The relabelling $P[f]$ acts like process $P$ with all labels $a$ replaced by $f(a)$.
Last, the rule for agent identifiers says that an agent $A$ has the same transitions as the body $P$ of its defining equation.

As usual, to avoid parentheses, we assume that the operators have decreasing binding strength in the following order:
restriction and relabelling, prefixing, parallel composition, choice.

The pair $\langle \T{CCS}, \rightarrow\rangle$ is called the \emph{labelled transition system} (LTS) of CCS\@.
\vspace{2ex}

\noindent
\hypertarget{Ex1}{\textbf{Example 1\;}}
We describe a simple shared memory system in CCS, using the name
\plat{$\ass{x}{v}$} for the assignment of value $v$ to the variable $x$, and 
$\noti{x}{v}$ for noticing or notifying that the variable $x$ has the value $v$.
The action $\overline{\ass{x}{v}}$ communicates the assignment $x:=v$ to the shared memory,
whereas $\ass{x}{v}$ is the action of the shared memory of accepting this communication.
Likewise, $\overline{\noti{x}{v}}$ is a notification by the shared memory that $x$ equals $v$; it synchronises
with the complementary action $\noti{x}{v}$ of noticing that $x=v$.

We consider the process
$(x^\tr\ParOp R\ParOp W)\backslash\{\ass{x}{\tr},\ass{x}{\fa}, \noti{x}{\tr},\noti{x}{\fa}\},$
where\vspace{-.8ex}
\[\begin{array}{r@{\ }c@{\ }l}
  x^\tr &\stackrel{\it def}{=}& \ass{x}{\tr}\mathbin.x^\tr \ +\ \ass{x}{\fa}\mathbin.x^\fa \ +\ \overline{\noti{x}{\tr}}\mathbin.x^\tr\ ,\\
  x^\fa &\stackrel{\it def}{=}& \ass{x}{\tr}\mathbin.x^\tr \ +\  \ass{x}{\fa}\mathbin.x^\fa \ +\  \overline{\noti{x}{\fa}}\mathbin.x^\fa\ ,\\
  R &\stackrel{\it def}{=}& \noti{x}{\tr}\mathbin.R\qquad\mbox{and}\qquad W \stackrel{\it def}{=} \overline{\ass{x}{\fa}}\mathbin.\textbf{0}\ .
\end{array}\]
\noindent
The processes $x^\tr$ and $x^\fa$ model the two states of a shared
Boolean variable~$x$ (\tr\ and \fa, respectively).
Both accept assignment actions, changing their state
accordingly. They also provide their
respective value to a potential reader.
The process $R$ (\textit{reader}) is an infinite loop which permanently tries to read
\begin{wrapfigure}[4]{r}{0.28\textwidth}
  \vspace{-1.5ex}
   \includegraphics[scale=1.5]{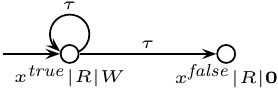}
\end{wrapfigure}
value \tr\ from variable $x$, and the process $W$ (\textit{writer}) tries once to set variable $x$ to false. 
Since the overall process uses the restriction operator, 
its transition system,  depicted on the right, has only two transitions, a \mbox{$\tau$-loop}
of $R$ reading the value, and a transition to $x^\fa|R|\nil$ of $W$ assigning $x$ to false---after that no
further transition is possible.
The justness assumption, to be described in Sects.~\ref{sec:unsatisfactory} and~\ref{sec:justness}, is not sufficient to ensure that the writer eventually performs its transition, and this fact is one of the motivations why we introduce signals in \autoref{sec:signals}.

\section{Peterson's Mutual Exclusion Protocol---Part I}\label{sec:Peterson}
In~\cite{GH15} it is shown that Peterson's mutual exclusion protocol~\cite{Peterson81} cannot be expressed in CCS without assuming fairness.
In this section we briefly recapitulate the protocol itself and
present an optimal rendering in CCS\@. In the next section we
discuss what the problems are with such a rendering.

The `classical' Peterson's mutual exclusion protocol deals with two concurrent processes \procA and \procB 
that want to alternate critical and noncritical
sections.

Each of the processes will stay only a finite amount of time in the critical section,
although it is allowed to stay forever in its noncritical section.
The purpose of the algorithm is to ensure that the processes are never
simultaneously in the critical section, and to guarantee that both
processes keep making progress;
in particular the latter means that if a process 
wants to access the critical section it will eventually do so.
 
A pseudocode rendering of Peterson's protocol is depicted in \autoref{fig:peterson}.
The processes use three shared variables: $\rA$, $\rB$ and $\tu$. 
The Boolean variable $\rA$ can be
written by Process $\procA$ and read by Process $\procB$, whereas $\rB$ can be
written by $\procB$ and read by $\procA$. By setting $\rA$ to $\tr$, Process
$\procA$ signals to Process $\procB$ that it wants to enter the critical
section. The variable $\tu$ can be written and read by both processes. 
Its carefully designed  functionality guarantees  mutual exclusion as well as deadlock-freedom. 
Both $\rA$ and $\rB$ are initialised with\/ $\fa$ and $\tu$ with $\it A$.
\begin{figure}[t]
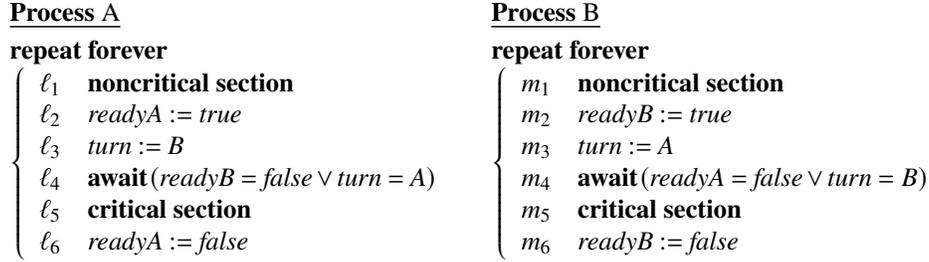

\centering
\small$
\begin{array}{@{}l@{}}
\underline{\bf Process~\procA}\\[.5ex]
{\bf repeat~forever}\\
\left\{\begin{array}{ll}
\ell_1 & {\bf noncritical~section}\\
\ell_2 & \it \rA := \tr	\\
\ell_3 & \it \tu := B \\
\ell_4 & {\bf await}\,(\it \rB = \fa \vee \tu = A) \\
\ell_5 & {\bf critical~section}\\
\ell_6 & \it \rA := \fa\\
\end{array}\right.
\end{array}
~~~~~~
\begin{array}{@{}l@{}}
\underline{\bf Process~\procB}\\[.5ex]
{\bf repeat~forever}\\
\left\{\begin{array}{ll}
m_1 & {\bf noncritical~section}\\
m_2 & \it \rB := \tr	\\
m_3 & \it \tu := A \\
m_4 & {\bf await}\,(\rA = \fa \vee \tu = B) \\
m_5 & {\bf critical~section}\\
m_6 & \it \rB := \fa \\
\end{array}\right.
\end{array}$
\vspace{-2mm}
\caption{Peterson's algorithm (pseudocode)}
\vspace{-2mm}
\label{fig:peterson}
\end{figure}

In order to model this protocol in CCS, we use the names
\ncA,~\crA,~\ncB, and \crB, for Processes \procA and \procB executing their (non)critical section.
The names $\ass{x}{v}$ and $\noti{x}{v}$ for the interactions of \procA and \procB with a shared memory
have been defined in \hyperlink{Ex1}{Ex.~1}.
The Processes \procA and \procB can be modelled~as
\[\begin{array}{@{}l@{~\stackrel{{\it def}}{=}~}l@{}}
\procA  & \ncA\mathbin.\overline{\ass{\rA}{\tr}}\mathbin.\overline{\ass{\tu}{\procB}}\mathbin.
     (\noti{\rB}{\fa}\ +\ \noti{\tu}{\procA})\mathbin.\crA\mathbin.\overline{\ass{\rA}{\fa}}\mathbin.\procA\ ,\\
\procB  & \ncB\mathbin.\overline{\ass{\rB}{\tr}}\mathbin.\overline{\ass{\tu}{\procA}}\mathbin.
     (\noti{\rA}{\fa}\ +\ \noti{\tu}{\procB})\mathbin.\crB\mathbin.\overline{\ass{\rB}{\fa}}\mathbin.\procB\ ,
\end{array}\]

\noindent where $(a+b).P$ is a shorthand for $a.P+b.P$.
This CCS rendering naturally captures the \textbf{await} statement, requiring Process \procA to wait
  at instruction $\ell_4$ until it can read that $\rB = \fa$ or $\tu = A$.
We use two agent identifiers for each Boolean variable, one for each value, similar to
\hyperlink{Ex1}{Ex.~1}. For example, we have
\plat{$\Tu{\procA}~\stackrel{{\it def}}{=}~\ass{\tu}{\procA}\mathbin.\Tu{\procA}\ +$}
\plat{$\ass{\tu}{\procB}\mathbin.\Tu{\procB} \ +\ \overline{\noti{\tu}{\procA}}\mathbin.\Tu{\procA}$.}
Peterson's algorithm is the parallel composition of all these processes,
restricting all the communications\vspace{-.3ex}
\[ (\procA \ParOp \procB \ParOp \RA{\fa} \ParOp \RB{\fa} \ParOp \Tu{A}) \backslash \RL\ ,\]

\vspace{-1mm}
\noindent where $\RL$ is the set of all names except \ncA,~\crA,~\ncB, and \crB.

It is well known that Peterson's protocol satisfies the safety
property that  both processes are never in the critical section at the same time.
In terms of  \autoref{fig:peterson}, there is no reachable state where
\procA and \procB have already executed lines $\ell_4$ and $m_4$ but have not
yet executed $\ell_6$ or $m_6$ \cite{Peterson81,Walker89}.
The validity of the liveness property, that any process
  leaving its noncritical section will eventually
enter the critical section, depends on its precise formalisation, as discussed in
Sects.~\ref{sec:unsatisfactory} and~\ref{sec:justness}.

\section{Why the CCS Rendering of Peterson's Algorithm is Unsatisfactory}\label{sec:unsatisfactory}

Liveness properties generally only hold under some assumptions.
The intended liveness property for Peterson's algorithm may already be violated if
both processes come to a permanent halt for no apparent reason. This behaviour should be considered
unrealistic. To rule it out
one usually makes a \textit{progress} assumption, formalised in
\autoref{sec:justness}, which can be formulated as follows \cite{TR13,GH14}:

\vspace{-0.15ex}
\begin{center}
\textit{%
Any process in a state that admits a non-blocking action will eventually perform an action.
}\end{center}

\vspace{-0.2ex}
Another example is an execution path $\pilive$ in which first Process $\procA$ completes instruction $\ell_1$;
leaving its noncritical section it implicitly wishes to enter the critical section.
Subsequently, Process~\procB cycles through its complete list of instructions in perpetuity
without \procA making
any further progress. This is possible because $\rA$ is never updated and always evaluates to false.
This execution path, if admitted, would be another counterexample to the
intended liveness property. However, progress is not sufficient to rule out such a path; after all
the whole system is making progress. To rule it out as a valid system run, we need the stronger assumption of
\emph{justness} \cite{GH14}, or an even stronger \emph{fairness} assumption \cite{Fr86}.

We formalise justness in the next section.
Here we sketch the general idea, and the difference with fairness, by an example.\pagebreak[3]

Suppose we have a vending machine with a single slot
for inserting coins, and there are two customers: one intends to insert an infinite supply of
quarters, and the other an infinite supply of dimes. No customer intends to ever extract something from
the vending machine. Since a quarter and a dime cannot be entered simultaneously, the two customers
compete for a shared resource. Should it be allowed for one customer to enter an unending sequence of quarters,
while the other does not even get in a single dime? The assumption of (strong or weak) fairness rules
out this realistic behaviour, while weaker assumptions like progress and justness allow this as one of the
valid ways such interactions between the customers and the vending machine may play out.

Alternatively, assume that the same two customers have access to a vending machine each, and that
each of these vending machines serves that customer only. In that case the assumption of progress is
not strong enough to rule out that one customer enters an unending sequence of quarters,
while the other does not even get in a single dime; after all the whole system keeps making progress at all times.
The assumption of justness guarantees that the customer will get a chance to enter his dimes by
applying the idea of progress to isolated components of a system; it
entails that the perpetual insertion of quarters by one customer in one machine in no way prevents
the other customer to insert dimes in the other machine.

In \cite{Walker89} Walker shows that once Process \procA executes instruction $\ell_2$, it will in
fact enter the critical section, i.e.\ execute $\ell_4$. This proof assumes progress, but not
justness, let alone fairness. The only question left is whether we can guarantee that execution of
$\ell_1$ is always followed by $\ell_2$. Thus, when assuming progress, the only possible
counterexample to the intended liveness property of Peterson's algorithm is the execution path
$\pilive$ sketched above, and its symmetric counterpart. This execution represents a battle for the
shared variable $\rA$. Process \procA tries to assign a value to this variable, whereas Process \procB engages
in an unending sequence of read actions of this variable (as part of infinitely many instructions $m_4$).
If we assume that the central memory in which variable $\rA$ is stored implements its own mutual
exclusion protocol, which prevents two processes from reading and writing the same variable at the same
time (but guarantees no liveness property), we may have the situation that Process \procA has to wait before setting $\rA$ to $\tr$ until
Process \procB is done reading this variable. However, Process \procB may be so quick that each time it is
done reading $\rA$ it executes $m_5$--$m_3$ in the blink of an eye and grabs hold of the same
variable for reading it again before Process \procA gets a chance to write to it.
Under this assumption we would conclude that Peterson's algorithm does not have the required
liveness property, since Process \procA may never get a chance to write to $\rA$ because Process \procB is
too busy reading it, and hence never ever enters the critical section.

However, it is reasonable to assume that in the intended setting where Peterson's algorithm would be
employed, the central memory does \emph{not} employ its own mutual exclusion protocol that prevents
one process from writing a variable while another is reading it.\footnote{Without such a protocol, it could
be argued that the reading process may read \emph{anything} when reading overlaps with 
writing the same variable. However, the variable $\rA$ has only two possible
values that can be read, and depending on which of these is returned, the overlapping read action
may just as well be thought to occur before or after the write action.}
With this view of the central memory in mind, instruction $\ell_2$ cannot be blocked by Process \procB,
and hence the assumption of justness is sufficient to ensure that Peterson's protocol \emph{does} have the
required liveness property.

The same conclusion cannot be drawn for the rendering of Peterson's algorithm in CCS\@.
Here the write action \plat{$\ell_2 = \overline{\ass{\rA}{\tr}}$} needs to synchronise with the action\vspace{-2pt}
$\ass{\rA}{\tr}$ of the shared memory storing variable $\rA$.\vspace{-2pt} That process has to make a choice
between executing $\ass{\rA}{\tr}$ and executing $\overline{\noti{\rA}{\fa}}$, the latter in synchronisation
with Process \procB\@. When it chooses $\overline{\noti{\rA}{\fa}}$ it has to wait until this instruction is
terminated before the same choice arises again. Hence the write action can be blocked by the
read action, and justness is not strong enough an assumption to ensure that eventually the
assignment will take place. Assuming fairness is of course enough to achieve this, but risky since it 
has the potential to rule out realistic behaviour (see above).

The above reasoning merely shows that the given implementation of Peterson's mutual exclusion
protocol in CCS requires fairness to be correct. In \cite{GH15} we show that the same holds
for any implementation of any mutual exclusion protocol in CCS, and the same argument applies to a
wider class of process algebras.
Peterson expressed his protocol in pseudocode without resorting to a fairness assumption. 
We understand that he assumes progress and justness
implicitly,  and 
accordingly his protocol and
liveness claim are correct.
It follows that Peterson's pseudocode does not admit an accurate translation into CCS.

\section{Formalising Progress and Justness}\label{sec:justness}

Liveness properties are naturally expressed as properties of
execution paths. A \emph{path} of a process $P$ is an alternating sequence
$P \ar{\alpha_1} P_1 \ar{\alpha_2} ... \ar{\alpha_{n}} P_n \ar{\alpha_{n+1}} ...$
of states and transitions.
A path can be finite or infinite. A possible formulation of the liveness
property  of Peterson's algorithm, applied to paths $\pi$, is that each occurrence of a transition labelled with $\ncA$
in $\pi$ is followed by an occurrence of $\crA$, and similarly for $\procB$.
To express when a liveness property holds for a system $P$, we need the notion of a
\emph{complete} path: one that describes a complete execution of $P$, rather than a partial one.
The property holds for $P$ iff it holds for all its complete paths.
Progress, justness and fairness assumptions rule out certain paths from being
considered complete---those that are in disagreement with the assumption.
The stronger the assumption, the more paths are ruled out, and the more likely
it is that a given liveness property holds.

A complete path ending in a state $P_n$ models a system run in which no further activity takes place 
after $P_n$ has been reached. 
A complete path ending in a transition models a system run where
transitions are considered to have a duration, and the final transition commenced, but never finishes.

One assumption we adopt in this paper is that ``atomic actions always terminate'' \cite{OL82}.
It rules out all paths ending in a transition.
To check whether Peterson's algorithm is compatible with this assumption, we note that
processes are not allowed to stay forever in their critical sections, so the actions $\crA$ and
$\crB$ can be assumed to terminate. Read and write actions of variables terminate as well.
However, a process is allowed to stay forever in its noncritical section, so the actions $\ncA$ and
$\ncB$ need not terminate. To make our formalisation of Peterson's algorithm compatible with the
assumption that actions terminate, we could split the action $\ncA$ into \textbf{start}$(\ncA)$
and \textbf{end}$(\ncA)$. Both these actions terminate, and an execution in which Process \procA stays in
its noncritical section corresponds with a complete path that ends in the state between these transitions.
To save the effort of rewriting the protocol from \autoref{sec:Peterson}, we shall
identify instruction $\ell_6$ with entering the noncritical section, and interpret $\ell_1$ as
leaving the noncritical section. Thus, the processes start out being in their noncritical sections.

To formalise the assumptions of progress and justness, we need the concept of a
\emph{non-blocking} action. A process of the form $\tau.P$ should
surely execute the internal action $\tau$,
and not stay forever in its initial state. However, a process $a.P$ running in the environment 
$(\_\ParOp E)\backslash \{a\}$ may very well stay in its initial state,\linebreak[2] namely when the environment $E$ never
provides a signal $\bar{a}$ that the process can read. With this in mind we assume a classification of the set of
actions into blocking and non-blocking actions \cite{GH14}.
The internal action $\tau$ is always non-blocking, and any action $a$ classified as 
non-blocking shall never be put in the scope of a restriction operator $\backslash \RL$ with $a\in \RL$,
and never be renamed into a blocking action \cite{GH15}.
The transition system of Peterson's algorithm features actions $\crA$, $\crB$, $\ncA$, $\ncB$ and
$\tau$---other names are forbidden by the restriction operator. We classify $\crA$, $\crB$ and $\tau$ 
 as non-blocking, but the actions $\ncA$ and $\ncB$ of leaving the noncritical sections may block.
Our progress assumption rules out as complete any path ending in a state in which a non-blocking
action is enabled, i.e.\ any system state except for the ones where both Processes \procA and \procB are
(back) in their initial state.

The (stronger) \emph{justness assumption} from \cite{GH14} is:
\begin{center}
\parbox{0.87\textwidth}{\textit{%
    If a combination of components in a parallel composition is in a state
    that admits a non-blocking action, then one (or more) of them  will eventually partake in an action.
    }}
\end{center}

\noindent
Its formalisation uses \emph{decomposition}:
a transition \plat{$P|Q \ar{\alpha} R$} derives, through the rules of \autoref{tab:CCSsos}, from
\begin{itemize}
\item a transition \plat{$P \ar{\alpha} P'$} and a state $Q$, where $R=P'|Q$\,,
\item two transitions \plat{$P \ar{a} P'$ and $Q \ar{\bar a} Q'$}, where $R=P'|Q'$ and $\alpha = \tau$\,,
\item or from a state $P$ and a transition \plat{$Q \ar{\alpha} Q'$}, where $R=P|Q'$.
\end{itemize}
This transition/state, transition/transition or state/transition pair is called a \emph{decomposition}
of \plat{$P|Q \ar{\alpha} R$}; it need not be unique.
A \emph{decomposition} of a path $\pi$ of $P|Q$ into paths $\pi_1$ and $\pi_2$ of
$P$ and $Q$, respectively, is obtained by \hypertarget{hr:decomp}{decomposing}\label{pg:decomp} each transition in the path, and
concatenating all left-projections into a path of $P$---the \emph{decomposition of $\pi$ along $P$}---and
all right-projections into a path of $Q$.
It could be that a path $\pi$ is infinite, yet either $\pi_1$ or $\pi_2$ (but not both) are finite.
Decomposition of paths need not be unique.

Similarly, any transition \plat{$P[f] \ar{\alpha} R$} stems from a transition \plat{$P \ar{\beta} P'$},
where $R=P'[f]$ and $\alpha=f(\beta)$.
This transition is called a decomposition of \plat{$P[f] \ar{\alpha} R$}. A \emph{decomposition}
of a path $\pi$ of $P[f]$ is obtained by decomposing each transition in the path, and
concatenating all transitions so obtained into a path of~$P$.
A decomposition of a path of $P\backslash \RL$ is defined likewise.

\begin{definition}\label{df:just path}\rm
The class of \emph{$Y\!$-just} paths, for $Y\mathbin\subseteq\HC$,
 is the largest class of paths in $\T{CCS}$ such that
\begin{itemize}
\item a finite $Y\!$-just path ends in a state that admits actions from $Y$ only;
\item a $Y\!$-just path of a process $P|Q$ can be decomposed into an $X$-just path of $P$ and a $Z$-just
  path of $Q$ such that $Y\mathbin\supseteq X\mathord\cup Z$ and $X\mathord\cap \bar{Z}\mathbin=\emptyset$---here
  $\bar Z\mathbin{:=}\{\bar{c} \mid c\mathbin\in Z\}$;
\item a $Y\!$-just path of
  $P\backslash \RL$ can be decomposed into a $Y\mathord\cup\RL\cup\bar\RL$-just path of $P$;
\item a $Y\!$-just path of $P[f]$ can be decomposed into an
  $f^{-1}(Y)$-just path of $P$; 
\item and each suffix of a $Y\!$-just path is $Y\!$-just.
\end{itemize}
A path $\pi$ is \emph{just} if it is $Y$-just for some set of blocking
actions $Y\subseteq\HC$.
A just path $\pi$ is \emph{$a$-enabled} for an action $a\in\HC$ if $a\in Y$ for all $Y$ such that $\pi$ is $Y\!$-just.
\end{definition}
Intuitively, a $Y\!$-just path models a run in which $Y$ is an upper~bound of the set of labels of
abstract transitions\footnote{The CCS process $a.0 | b.0$ has two transitions labelled $a$, namely
\plat{$a.0 | b.0 \ar{a} 0|b.0$} and \plat{$a.0 | 0 \ar{a} 0|0$}. The only difference between these two transitions
  is that one occurs before the action $b$ is performed by the parallel component and the other
  afterwards. In \cite{GH14} we formalise a notion of an \emph{abstract transition} that
  identifies these two concrete transitions.} that from some point onwards are continuously enabled but never taken.
Here an {abstract transition} with a label from $\HC$ is deemed to be continuously enabled but never
taken iff it is enabled in a parallel component that performs no further actions.
Such a run can occur in the modelled system if the environment from some point
onwards blocks the actions in $Y$.

Now consider the path $\pilive$
violating the intended
liveness property. The decompositions of $\pilive$ along Processes \procA and \procB were mentioned in
\autoref{sec:unsatisfactory}.
These paths are $\{\, \overline{\ass{\rA}{\tr}}\, \}$-just and $\emptyset$-just, respectively.
The decomposition along $\Tu{A}$ is an infinite path
taking action $\ass{\tu}{A}$ only ($\emptyset$-just). The decomposition along $\RB{\fa}$ is an
infinite path
alternatingly taking actions $\ass{\rB}{\tr}$ and $\ass{\rB}{\fa}$ (also $\emptyset$-just) 
and the decomposition along $\RA{\fa}$ is an infinite path
taking action $\overline{\noti{\rA}{\fa}}$ only (again $\emptyset$-just).
It follows that the composition $\pilive$ of these five paths is $\emptyset$-just.
Intuitively this is the case because no communication is permanently enabled and never
taken.\vspace{-4pt} In particular, the communication $\overline{\ass{\rA}{\tr}}$ is
disabled each time the component $\RA{\fa}$ does the action $\overline{\noti{\rA}{\fa}}$ instead.

\section{CCS with Signals}\label{sec:signals}
We would like to prevent such a path to be complete in a CCS model of Peterson's algorithm.
In order to achieve this, we propose to replace an
action such as $\overline{\noti{\rB}{\fa}}$, which makes the variable busy even if it is only read,
by a state predicate providing its value. This
mode of communication is called \emph{signalling}.

\emph{CCS with signals} (\CCSs) is CCS extended with a signalling operator. 
Informally, the signalling operator $P\signals s$ emits the signal $s$ to be read by another process. 
Signal emission cannot block other actions.
Formally, CCS is extended with a set $\Sig$ of \emph{signals},
ranged over by $s,t,\dots$.
In \CCSs the set of actions is defined as $Act := \Sig \dcup \HC\dcup\{\tau\}$.
A relabelling is a function $f:(\Sig\rightarrow\Sig)\cup(\HC\rightarrow\HC)$ satisfying 
$f(\bar a) = \overline{f(a)}$.
As before it extends to $Act$ by $f(\tau) = \tau$.

The class $\T{\CCSs}$ of \emph{\CCSs expressions} is defined as the smallest class that includes

\vspace{-2ex}
\begin{multicols}{3}
    \begin{itemize}
        \item \emph{agent identifiers} $A\in\K$\,;
        \item \emph{prefixes} $\alpha.P$\,;
        \item \emph{(infinite) choices} $\sum_{i\in I} P_i$\,;
        \item \emph{parallel compositions} $P|Q$\,;
        \item \emph{restrictions} $P\backslash \RL$\,; 
        \item \emph{relabellings} $P[f]$\,;
         \item \emph{signallings} $P\signals s$
    \end{itemize}
 \end{multicols}

 \vspace{-2ex}
 \noindent where $P,P_{i},Q\in\T{\CCSs}$ are \CCSs expressions, $I$ an index set,
 $\RL \subseteq \A \cup \Sig$ a set of handshake names and signals, $f$ an arbitrary relabelling function,
 and $s\in \Sig$ a signal. The new operator $\ \signals\ $ binds as strong as relabelling and restriction.

\newcommand{\emptySig}{\ensuremath{\emptyset}}
\newcommand{\sigar}[1]{\ar{#1}\emptySig}
\renewcommand{\sigar}[1]{^{\curvearrowright #1}}
\renewcommand{\P}{\mathcal P}

The semantics of \CCSs is given by the labelled transition relation
$ \mathord\rightarrow \subseteq \T{\CCSs}\times Act \times \T{\CCSs}$
and a predicate $\sigar{}\subseteq \T{\CCSs}\times \Sig$ that
are derived from the rules of CCS (\autoref{tab:CCSsos}, where $\alpha$ can also be a signal)
and the new rules of \autoref{tab:CCSssos}.
\begin{table}[t]
\normalsize
\begin{center}
\framebox{$\begin{array}{cccc}
(P \signals s)\sigar{s}&
\displaystyle\frac{P\ar{\alpha}P'}{P\signals s \ar{\alpha}P'}&
\displaystyle\frac{P\sigar{s}}{(P\signals t)\sigar{s}}&
\displaystyle\qquad\frac{{P_j}\sigar{s}}{(\sum_{i\in I}P_i) \sigar{s}}~~(j\mathbin\in I)\\[4ex]
\displaystyle\frac{P\sigar{s}}{(P|Q) \sigar{s}} &
\displaystyle\frac{P\sigar{s},~ Q \ar{s} Q'}{P|Q \ar{\tau} P| Q'} &
\displaystyle\frac{P\ar{s}P',~ Q \sigar{s}}{P|Q \ar{\tau} P'| Q}&
\displaystyle\frac{Q \sigar{s}}{(P|Q) \sigar{s}}\\[4ex]
\displaystyle\frac{P \sigar{s}}{(P\backslash \RL) \sigar{s}}~~(s\not\in\RL) &
\displaystyle\frac{P \sigar{s}}{P[f] \sigar{f(s)}} &
\displaystyle\frac{P \sigar{s}}{A\sigar{s}}~~(A \stackrel{{\it def}}{=} P)
\end{array}$}
\end{center}
\vspace{-3ex}
\caption{Structural operational semantics for signals of \CCSs}
\label{tab:CCSssos}
\end{table}
The predicate $P\sigar{s}$ indicates that process $P$ emits the signal $s$,
  whereas a transition $P \ar{s} P'$ indicates that
  $P$ reads the signal $s$ and thereby turns into $P'$.
The first rule is the base case showing that
a process $P\signals s$ emits the signal $s$.
The second rule models the fact that signalling cannot prevent a process from making progress. 
After having taken an action, the signalling process loses its ability to emit the signal.
It is essentially this rule which fixes the read/write problem presented in the previous section.
The two rules in the middle of \autoref{tab:CCSssos} state that the
action of reading a signal by one component in
(parallel) composition together with the emission of the same signal by another component,
results
in an internal transition $\tau$; similar to the case of handshake communication.
Note that the component emitting the signal does not change through this interaction.
All the other rules
of \autoref{tab:CCSssos} lift the emission of $s$
by a subprocess $P$ to the 
overall process.
\autoref{tab:CCSssos} can easily be adapted to other process calculi, hence our extension is not limited to CCS.

We give an example similar to the one at the end of \autoref{sec:ccs} to illustrate the use of signals. 
\vspace{2ex}

\noindent
\hypertarget{Ex2}{\textbf{Example 2\; }}
We  describe once again a one-variable shared memory system with an infinite reader $R$ and a single
writer~$W$. But this time communication actions $\noti{x}{v}$ and
$\overline{\noti{x}{v}}$\vspace{2pt} are replaced with signals $\noti{x}{v}$. The variable $x$
now emits a signal notifying its value, so we have: 
\plat{$x^\tr \stackrel{\it def}{=} (\ass{x}{\tr}\mathbin.x^\tr \, +\, \ass{x}{\fa}\mathbin.x^\fa) \signals \noti{x}{\tr}$} and 
\plat{$x^\fa \stackrel{\it def}{=} (\ass{x}{\tr}\mathbin.x^\tr \,+\, \ass{x}{\fa}\mathbin.x^\fa) \signals \noti{x}{\fa}$}; the rest of the example remains unchanged.
The transition system is exactly the same, but now justness guarantees that the variable $x$ will eventually be set to $\fa$. 
This is in contrast to \hyperlink{Ex1}{Ex.~1}, where
  it is not guaranteed that $x$ will eventually be $\fa$, even when assuming justness.
More precisely, if only the reader takes actions, $x^\tr$ is
now not progressing because it is emitting a signal
only, and then, assuming justness, it must eventually enter into communication with the writer.
\vspace{2ex}

As we have extended CCS with a novel operator, we have to make sure that our extension behaves `naturally', in the way one would expect.
\begin{theorem}
Strong bisimilarity \cite{vG00} is a congruence for all operators of \CCSs.
\end{theorem}
\begin{proof}
We get the result directly from the existing theory on structural operational semantics, as a result of carefully designing our language.
All rules of Tables~\ref{tab:CCSsos} and \ref{tab:CCSssos} are 
in the \emph{path} format of Baeten and Verhoef~\cite{BV93}, and hence the theorem holds~\cite{BV93}.
\end{proof}

\begin{theorem}\label{thm:associativity}
The operator $|$ is associative and commutative,
and the operator $\ \signals\ $ is pseudo-commutative, i.e.\ $P\signals s \signals t = P\signals t \signals s$,
all up to bisimilarity.
\end{theorem}
\begin{proof}\hspace{-1pt}
Our process algebra with predicates can easily be encoded in a process algebra without, by
writing
\[P \ar{\bar s} P\mbox{ for }P\sigar{s}\ .\] 
On the level of the structural operational semantics,
this amounts to letting $\alpha$ range over $Act \cup \{\bar s \mid s \in \Sig\}$ in the rules of
\autoref{tab:CCSsos}, and changing the first rule of \autoref{tab:CCSssos} into
$P \signals s \ar{\bar s} P \signals s$.
The third rule of \autoref{tab:CCSssos} becomes an instance of the second (with
$\alpha \in Act \cup \{\bar s \mid s \in \Sig\}$), and the remaining rules of 
\autoref{tab:CCSssos} become special cases of the rules of \autoref{tab:CCSsos}.

Clearly, two processes are bisimilar in the original \CCSs iff they are bisimilar in this encoding.
Since the parallel composition of the encoded \CCSs is the same as the one of CCS,
it is known to be associative and commutative up to bisimilarity \cite{Mi89}.

To prove pseudo-commutativity of $\ \signals\ $, we note that $P
\signals s \signals t$ and $P \signals t \signals s$ have exactly the same outgoing transitions and
signals, thereby being trivially equal up to bisimilarity.\vspace{2ex}
\end{proof}

Since we extended CCS, we also have to extend our definition of justness. The decomposition of paths
remains unchanged, except that a 
transition \plat{$P|Q \ar{\tau} R$} can now derive, through the rules of \autoref{tab:CCSssos}, from
signal communication.
In that case we consider the decomposition along the signalling process empty, just as if it was an application of the left- or right-parallel composition rule. Because processes can communicate through signalling, we first introduce the definition of signalling paths. Informally, a path emits signal $s$ if one component in the parallel composition ends in a state where signal~$s$ is activated.
A $Y\!$-signalling path is a path where $Y$ is an upper bound on the signals emitted by the path.

\begin{definition} \rm
The class of \emph{$Y\!$-signalling} paths, for $Y\mathbin\subseteq\Sig$,
 is the largest class of paths in $\T{CCS}$ such that
\begin{itemize}
	\item a finite $Y\!$-signalling path ends in a state that admits signals from $Y$ only;	
	\item a $Y\!$-signalling path of a process $P\ParOp Q$ can be decomposed into an $X$-signalling path 
	of $P$ and a $Z$-signalling path of $Q$ such that $Y\mathbin\supseteq X\mathord\cup Z$;
	\item a $Y\!$-signalling path of $P\backslash \RL$ can be decomposed into a $Y\mathord\cup\RL_\Sig$-signalling path of $P$%
	---here $\RL_\Sig\mathbin{:=}\RL \cap \Sig$ restricts the set $\RL$ to signals;
\item a $Y\!$-signalling path of $P[f]$ can be decomposed into an
  $f^{-1}(Y)$-signalling path of $P$; 
	\item and each suffix of a $Y\!$-signalling path is $Y\!$-signalling.
\end{itemize}
\end{definition}

\noindent Using this definition, we can adapt the definition of justness of \autoref{sec:justness}.\pagebreak[3]

\begin{definition} \rm\label{def:justness}
The class of \emph{$Y\!$-just} paths, for $Y\mathbin\subseteq\HC\cup\Sig\!\!$,
is the largest class of paths in $\T{\CCSs}$ such that
\begin{itemize}
	\item a finite $Y\!$-just path ends in a state that admits actions from $Y$ only;
	\item a $Y\!$-just path of a process $P | Q$ can be decomposed into a path of $P$ that is
          $X$-just and $X'$-signalling, and a path of $Q$ that is $Z$-just and $Z'$-signalling,
  such that $Y\mathbin\supseteq X\mathord\cup Z$, $X \cap \bar Z_\HC = \emptyset$, $X \cap Z' = \emptyset$ and $X' \cap Z = \emptyset$%
  ---here $\bar Z_\HC := \{\bar a \mid a\in Z \cap\HC\}$;
  \item a $Y$-just path of  $P\backslash\RL$ can be decomposed into a $Y \cup \RL \cup \bar L_\HC$-just path of $P$;
  \item a $Y\!$-just path of $P[f]$ can be decomposed into an
  $f^{-1}(Y)$-just path of $P$;
\item and each suffix of a $Y\!$-just path is $Y\!$-just.
\end{itemize}
As before, a path $\pi$ is \emph{just} if it is $Y$-just for some
set of blocking actions and signals $Y\subseteq\HC \cup \Sig\!$.
A just path $\pi$ is \emph{$a$-enabled} for $a\in\HC \cup \Sig$ if $a\in Y$ for all $Y$ such that
$\pi$ is $Y\!$-just.
\end{definition}

\noindent
The condition on signals in the second item guarantees that a process $(\textbf{0}\signals s \ParOp s.\nil) \backslash \{ s \}$ makes progress.

\newcommand{\rholive}{\rho_R}

The encoding in the proof of \autoref{thm:associativity} does not preserve justness.
In \hyperlink{Ex2}{Ex.~2}, for instance, applying the operational
  semantics of Table~\ref{tab:CCSssos},
the path $\rholive$ involving infinitely many read
actions but no write action is not just, because its decomposition along $x^\tr$ is finite and
$\ass{x}{\fa}$-enabled, whereas its decomposition along $W$ is $\overline{\ass{x}{\fa}}$-enabled; so
by the second clause of Def.~\ref{def:justness} $\rholive$ is not just: there are no $Y$, $X$ and $Z$
such that the condition \mbox{$X \cap \bar Z_\HC = \emptyset$} is satisfied.
Yet, after applying the encoding in the proof of \autoref{thm:associativity}, the decomposition
along $x^\tr$ becomes infinite and $\emptyset$-just, and $\rholive$ becomes just. 
This is the main reason we did not present the semantics of \CCSs in this form from the onset.

\section{Peterson's Mutual Exclusion Protocol---Part II\label{sec:peterson2}}
We now present an implementation of Peterson's mutual exclusion algorithm in \CCSs. We
use the same notation as in \autoref{sec:Peterson}, except that
actions $\noti{x}{v}$ and $\overline{\noti{x}{v}}$ are replaced with signals $\noti{x}{v}$, just as in \hyperlink{Ex2}{Ex.~2}.
Only the variable processes change, such as
\plat{${\Tu{\procA}}  \stackrel{\it def}{=} (\ass{\tu}{\procA}\mathbin.\Tu{\procA}\, +\, \ass{\tu}{\procB}\mathbin.\Tu{\procB}) \signals {\noti{\tu}{\procA}}$}; Processes $\procA$ and $\procB$ are unchanged. The protocol rendering is still 
$(\procA \ParOp \procB \ParOp \RA{\fa} \ParOp \RB{\fa} \ParOp \Tu{A}) \backslash \RL$,
where $\RL$ is the set of all names and signals except \ncA,~\crA,~\ncB, and \crB, as before.

In the remainder of this section we prove Peterson's protocol correct, i.e.\ safe and live. 
We include the proof of safety for completeness, but concentrate on liveness.

\begin{theorem}
Peterson's 
protocol is safe.
In terms of  \autoref{fig:peterson}, there is no reachable state where
\procA and \procB have already executed lines $\ell_4$ and $m_4$ but have not
yet executed $\ell_6$ or $m_6$.
\end{theorem}

\begin{proof}
We follow the proof by contradiction of Peterson~\cite{Peterson81}.
Suppose both processes succeed the test at $\ell_4$ and $m_4$.
Let $\procA$ be the first to pass this test. At that time
either $\rB$ was
false (meaning that Process~$\procB$ was between $m_6$ and $m_2$)
or $turn$ was set to $A$.
In the first case, $\rA$ will not be set to false before Process $\procA$ leaves the critical section
and $\tu$ is bound to be set to $A$ by Process $\procB$ before $m_4$ is executed.
So the test at $m_4$ will fail. In the
second case, since $\procA$ is about to enter the critical section, $\tu$ cannot be set to $B$ anymore
and $\rA$ is \tr, so again the test $m_4$ will fail for Process $\procB$. 
\vspace{2ex}
\end{proof}

\noindent
Peterson's protocol satisfies also the liveness property. As mentioned before, this result
could not be proven for the formalisation of the protocol in CCS, assuming justness only.
\begin{theorem}
Assuming justness, Peterson's protocol satisfies the liveness property:
on each just path, each occurrence of $\ncA$ is followed by $\crA$ (and similarly for $\procB$).
\end{theorem}

\begin{proof}
Let $\pi$ be a just path of the protocol. Since $\ncA$ and $\ncB$ are the only
possible blocking actions, $\pi$ must be $\{\ncA,\ncB\}$-just.
If we get rid of all the restrictions we
obtain a $Y$-just path of $(\procA \mid \procB \mid \RA{\fa} \mid \RB{\fa} \mid \Tu{\procA})$ where
$Y = \{\ncA,\,\ncB\} \cup \RL \cup \bar\RL_\HC$. Suppose its decomposition $\pi_A$ along
Process $\procA$ ends somewhere between
instructions $\ell_1$ and $\ell_4$. Then the decomposition $\pi_{\RA{}}$ along $\RA{\fa}$ is
also finite since only $\procA$ can communicate with this process. 
Using the CCS rendering from \autoref{sec:Peterson} this statement would be incorrect, since
there Process $\procB$ can  constantly interact with $\RA{\fa}$,
by reading its value; resulting in an infinite path \expandafter\ifx\csname graph\endcsname\relax
   \csname newbox\expandafter\endcsname\csname graph\endcsname
\fi
\ifx\graphtemp\undefined
  \csname newdimen\endcsname\graphtemp
\fi
\expandafter\setbox\csname graph\endcsname
 =\vtop{\vskip 0pt\hbox{%
\pdfliteral{
q [] 0 d 1 J 1 j
0.576 w
0.072 w
q 0 g
11.16 -15.048 m
5.256 -10.512 l
12.6 -11.736 l
11.16 -15.048 l
B Q
0.576 w
5.256 -6.192 m
13.392 -2.7 l
18.92448 -0.32544 24.48864 0.792 30.78 0.792 c
37.07136 0.792 40.032 -2.1456 40.032 -8.388 c
40.032 -14.6304 37.07136 -17.568 30.78 -17.568 c
24.48864 -17.568 19.02816 -16.48512 13.716 -14.184 c
5.904 -10.8 l
S
Q
}%
    \graphtemp=.5ex
    \advance\graphtemp by 0.116in
    \rlap{\kern 0.812in\lower\graphtemp\hbox to 0pt{\hss $\noti{\rA}{\fa}$\hss}}%
    \hbox{\vrule depth0.232in width0pt height 0pt}%
    \kern 0.812in
  }%
}%

$\RA{\fa}\!$\mbox{\raisebox{2.4ex}{\box\graph}}\hfill.

By Def.~\ref{def:justness}, the path $\pi_A$ must be $X$-just, and the path $\pi_{\RA{}}$ $Z$-just,
  for sets $X,\,Z \subseteq Y$ with $X \cap \bar Z_\HC = \emptyset$.
  Furthermore, $\ass{\rA}{\tr} \in Z$, since this action is enabled in the
  last state of $\pi_{\RA{}}$. Hence $\overline{\ass{\rA}{\tr}} \notin X$.
  Therefore $\pi_A$ cannot end right before instruction $\ell_2$.
As a result,
Process $\procA$ is stuck either right before $\ell_3$, or right before $\ell_4$.
In both cases Process $\procB$ would not be able to pass the test before the critical
section more than once. Indeed, in either case $\rA$ is
already set to \tr, thus Process $\procB$ must use $\tu=B$ to enter its
critical section. But, if trying to enter a second time, it would be
forced to set $\tu$ to $\procA$ and will be stuck.
When Processes $\procA$ and $\procB$ are both stuck, the path $\pi$
is finite and an action $\tau$ or $\ncA$ stemming from instruction $\ell_3$ or $\ell_4$ is enabled
at the end, contradicting, through the first clause of Def.~\ref{def:justness}, the
$\{\ncA,\ncB\}$-justness of $\pi$.
\end{proof}

\section[Peterson's Algorithm for N Processes]{Peterson's Algorithm for $N$ Processes}\label{sec:petersonn}

In the previous section we presented an implementation in \CCSs of Peterson's algorithm of mutual
exclusion for two processes. In \cite{Peterson81}, Peterson also presents a generalisation of his mutual
exclusion protocol to $N$ processes. In this section we describe the
algorithm and explain which assumptions should be made on the memory model in order for this protocol to
be correct, for $N {>}2$.
We claim that these assumptions are somewhat unrealistic.
\begin{figure}[b]
\vspace{-6pt}
\small
\[
\begin{array}{@{}l@{}}
\underline{\bf Process~i}~~(i\in\{1,\dots,N\})\\[1ex]
{\bf repeat~forever}\\
\left\{\begin{array}{ll}
\ell_1 & {\bf noncritical~section}\\
\ell_2 & {\bf for\ } j {\bf\ in\ } 1\dots N-1\\
\ell_3 & \quad \it room[i] := j\\
\ell_4 & \quad \it last[j] := i\\
\ell_5 & \quad {\bf await}\,(\it last[j] \neq i \vee (\forall k \neq i,~ room[k] < j))\\
\ell_6 & {\bf critical~section}\\
\ell_7 & \it room[i] := 0\vspace{-2pt}
\end{array}\right.
\end{array}
\]
\caption{Peterson's algorithm for $N$ processes (pseudocode)}\vspace{-2pt}
\label{fig:petersonN}
\end{figure}

\newcommand{\room}{Room}%
A pseudocode rendering of Peterson's protocol is depicted in \autoref{fig:petersonN}.
In order to proceed to the critical section, each process must go through $N{-}1$ locks
(\textit{rooms}). The shared variable $\it room[i] = j$ indicates that process number $i$ is currently
in \room~$j$. The shared variable $\it last[j] = i$ indicates that the last process to
`\textit{enter}' \room~$j$ is Process $i$. A process can go to the next room if and only if it
is not the last one to have entered the room, or if all other processes are strictly behind it. This
algorithm is also called the \emph{filter lock} because it ensures that for all $j$, no more than
$N{+}1{-}j$ processes are in rooms greater or equal than~$j$.
The critical section can be thought of as \room~$N$.

A natural memory model, used in \cite{bakery}, stipulates that memory accesses from different
  components can overlap in time, and that a 
  read action that overlaps with a write action of the same variable
  may yield \emph{any} value. Extending this idea, 
  we assume that when two concurrent write actions overlap, \emph{any} possible value
  could end up in the memory. We argue that the algorithm fails to satisfy mutual
  exclusion when assuming such a model.

  Suppose there are three processes, $\procA$, $\procB$ and $\procC$, and Processes $\procA$ and $\procB$ execute
  $\ell_1$--$\ell_4$ more or less simultaneously. When their instructions $\ell_4$ overlap, the value
  $\procC$ ends up in the variable $\it last[1]$---or any other value different from $\procA$ and $\procB$. Hence they both perform $\ell_5$, as well as
  $\ell_3$--$\ell_4$ for $j{=}2$. Again, the value $\procC$ ends up in $\it last[2]$.
  Subsequently, they both enter their critical section, and disaster strikes.

It follows that Peterson's algorithm for $N{>}2$ only works when running on a memory where write
  actions cannot overlap in time, or---if they do---their effect is the same as when one occurred
  before the other. Such a memory can be implemented by having a small hardware lock around a write
  action to the same variable. This entails that one write action would have to wait until the other one is completed.
  A memory model of this kind is implicitly assumed in process algebras like CCS(S).

  We show that, under such a memory model, Peterson's algorithm for $N{>}2$ does not satisfy
  liveness, unless we enrich it with an additional fairness assumption.

To prove this statement, let $N = 3$ and call the processes $\procA$, $\procB$ and $\procC$.
We show that (without the additional assumption) 
Process $\procA$ can be stuck at $\ell_4$ for $j {=} 1$.
Suppose Process $\procA$ is at this line. Then $\procA$ is about to set $\it last[1]$ to $A$, but has not written yet.
We can imagine the following scenario: Process $\procB$ enters \room~1, and sets $\it last[1]$ to $B$;
then Process $\procC$ enters \room~1, and sets $\it last[1]$ to $C$. This allows $\procB$ to
proceed to \room~2, then to go in the critical section (because all other processes are
still in room~1), and  then to  go back to \room~1, setting $\it last[1]$ to $B$. This allows
$\procC$ to go to \room~2, to the critical section, and back to room 1, setting $\it last[1]$ to $C$. Next $\procB$ 
can enter the critical section again, etc. 
Hence Processes $\procB$ and $\procC$ can go alternately in the critical
section without giving $\procA$ a chance to set variable $\it last[1]$. (The  variable
is too busy being written by $B$ and $C$.) This scenario cannot happen for $N = 2$ because
after $\procB$ sets $\it last[1]$ to $B$, $\procB$ is blocked until $\procA$ sets it to $A$; so
$\ell_4$ will eventually happen (with progress as a basic assumption).

As a consequence, in order for Peterson's algorithm to be live for more than two processes, we must adopt 
the additional fairness assumption that \emph{if a process permanently tries to write to a variable, it will eventually do so}, even if other processes are competing for writing to the same variable.
This property appears to be at odds with having a hardware lock around the shared variable.
Moreover, it cannot be implemented in \CCSs assuming only justness: when two
competitive processes try to write the same variable, nothing guarantees that both will eventually
succeed.\footnote{Let us consider a \CCSs process $(x^\tr | W_1 | W_2)\backslash \RL$ where processes
$W_1$ and $W_2$ are infinite writers ($W_i \stackrel{\it def}{=} \overline{\ass{x}{\fa}}.W_i$) and $L$ is the set of communication names. A path where $W_1$ always succeeds,\vspace{-2pt}
meaning that the decomposition along $W_2$ is empty, is just because the latter decomposition is 
$\{\,\overline{\ass{x}{\fa}}\,\}$-just and all the other decompositions $\emptyset$-just.} As a result
any \CCSs-rendering of Peterson's algorithm for $N$ processes does not possess the liveness property,
unless one makes a fairness assumption. The problem comes from the fact that
the variables $\it last[\cdot]$ are written by several parallel processes.
Signals only allow a writer to set a variable while it is being read but do not allow multiple writers at the same time.

We believe that the problem does not come from a lack of expressiveness of \CCSs but from
the protocol, which, while not being incorrect in itself, requires a memory model that 
assumes write actions to happen eventually, even though simultaneous
    interfering write actions are excluded; whether this is a realistic assumption on modern
    hardware requires further investigation.
\vspace{-0.09pt}

\section{Lamport's Bakery Algorithm}

In this section we analyse Lamport's bakery algorithm \cite{bakery},  another
mutual exclusion protocol for $N$ processes. It has the property that
processes write to separate variables; only the read actions are shared. We give a model for this
algorithm in \CCSs and prove its liveness property, assuming justness only.
\begin{figure}
\small
\[
\begin{array}{@{}l@{}}
\underline{\bf Process~i}~~(i\in\{1,\dots,N\})\\[1ex]
{\bf repeat~forever}\\
\left\{\begin{array}{ll}
\ell_1 & {\bf noncritical~section}\\
\ell_2 & \it choosing[i] := \tr\\
\ell_3 & \it number[i] := 1 + \max(number[1],\dots,number[N]);\\
\ell_{4} & \it choosing[i] := \fa\\
\ell_{5} & {\bf for\ }j {\bf\ in\ } 1\dots N\\
\ell_{6} & \quad {\bf await}\,(\it choosing[j] = \fa)\\
\ell_{7} & \quad {\bf await}\,(\it number[j] = 0 \vee (number[i],i) \leq (number[j],j))\\
\ell_{8} & {\bf critical~section}\\
\ell_{9} & \it number[i] := 0
\end{array}\right.
\end{array}
\]
\caption{Lamport's bakery algorithm for $N$ processes (pseudocode)}
\label{fig:bakery}
\vspace{-1ex}
\end{figure}

A pseudocode rendering of Lamport's bakery algorithm is depicted in \autoref{fig:bakery}.
Lines $2$--$4$ are called the \textit{doorway} and lines $5$--$7$ are called the \textit{bakery}.
In the doorway each process `takes a ticket' that has a number
strictly greater than all the numbers from the
other processes (at the time the process reads them). 
The variable $\it choosing[i]$ is a lock that makes line $3$, which is usually
  implemented by a simple loop, more or less `atomic'.
To ensure that the holder of the lowest number is next in the critical section, 
each process goes through a number of
locks in the bakery (Lines $5$--$7$). 
When  process $i$ enters the critical section, the value it has read for
$\it number[j]$, if not $0$, is greater or equal than its own $\it number[i]$,  for all $j$.

We now model this algorithm in CCSS\@. As usual, we define one agent for every pair
$(\mbox{variable}, \mbox{value})$. The variables $\it choosing$ can take values \tr\ or \fa, and $number$ any
non-negative value. The modelling of a Boolean variable is addressed in \hyperlink{Ex1}{Ex.~1},
and for the integer variables we define:

\[ \it number[i]^k \stackrel{def}{=} \Big(
\sum_{l \in \mathbb N} \ass{\it number[i]}{l} \mathbin. \it
number[i]^l \Big)
\signals \noti{\it number[i]}{k} \;.\]
\noindent
Each process
 $i$ begins with a non-critical section before entering the doorway.
\[ P_i \stackrel{\it def}{=} \textbf{noncrit[$i$]}\mathbin.\overline{\ass{\it choosing[i]}{\tr}} \mathbin. \it doorway[i]_0^1 \]

 Line $3$ encodes several read actions, an arithmetic operation, and an assignment in a single step. 
 In CCS(S) (and most programming languages) this command is modelled by several atomic steps, e.g.\ 
 by the simple loop $m := 0;\; {\bf for\ } j {\bf\ in\ } 1\dots N\{ m := \texttt{max(} m, number[j] \texttt{)}\};\;
{\it number[i]} := 1+m$.
We define processes
$\it doorway[i]_m^j$ that represent the state of being in the doorway \textbf{for}-loop for a process $i$ with loop index $j$ and local variable $m$ 
storing the current maximum.
\[ \it doorway[i]_m^j \stackrel{def}{=} 
\Big(\sum_{k > m} \noti{\it number[j]}{k} \mathbin.\it
doorway[i]_k^{j+1}\Big)
+ \Big(\sum_{k \leq m} \noti{\it number[j]}{k} \mathbin.\it doorway[i]_m^{j+1}\Big) \ ,\ j \in \{1,\dots,N\}\]
We then define $\it doorway[i]_m^{N+1}$, which represents
the termination of the \textbf{for}-loop by
\[ \it doorway[i]_m^{N+1} \stackrel{def}{=} \overline{\ass{number[i]}{m+1}}\mathbin. \overline{\ass{choosing[i]}{\fa}} \mathbin. bakery[i]_{m+1}^1 \;.\]
The process $\it bakery[i]_m^j$ represents the state of being in the bakery \textbf{for}-loop for process $i$
with loop index $j$ and $\it number[i] = m$. For $j \in \{1,\dots,N\}$:
\[ \it bakery[i]_m^j  \stackrel{def}{=} \noti{\it choosing[j]}{\fa} \mathbin. 
\Big(\noti{\it number[j]}{0} \ + \sum_{k > m \vee (k = m \wedge j \geq i)} \noti{number[j]}{k}
\Big)
\mathbin.bakery[i]_m^{j+1} \;.\]
Finally, $\it bakery[i]_m^{N+1}$ is the exit of the bakery \textbf{for}-loop, granting access to the critical section:\vspace{-2pt}
\[\it bakery[i]_m^{N+1} = \textbf{crit[$i$]} \mathbin.\overline{\ass{\it number[i]}{0}} \mathbin.P_i \]

\noindent
Our bakery algorithm is the parallel composition of all 
processes $P_i$, in combination with the shared variables $\it choosing[i]$ and $\it number[i]$,
restricting the communication actions:
\[ 
\Big(\BigParOp_{i \in \{1,\dots,N\}}
(P_i \mid \it choosing[i]^{\fa} \mid number[i]^0) 
\Big)\backslash \RL \ , \]
where $\RL$ is the set of all names and signals except \textbf{noncrit[$i$]} and \textbf{crit[$i$]}.

We now prove the liveness of (our rendering of) the algorithm, given that it is straightforward to adapt Lamport's proof of safety of the pseudocode \cite{bakery} to \CCSs. Since every process writes in its own variables, no process can be stuck because of concurrent writing.
Therefore, the only possibility for a process (call it \procA) to be stuck is
when trying to read a variable, so at $\ell_{3}$, $\ell_{6}$ or $\ell_{7}$.

If Process \procA is stuck at $\ell_{3}$, trying to read $\it number[\procB]$ for some process $\procB$,
  $\procB$ will get stuck at $\ell_{6}$ for $j{=}\procA$, because $\it choosing[\procA]$ remains \fa.
  So, Process $\procB$ cannot be perpetually busy writing $\it number[\procB]$, and $\procA$ cannot be
  stuck at $\ell_{3}$.

If $\procA$ is stuck at $\ell_{6}$, then  from the point of
view of \procA, some process \procB is all the time in  the doorway.
It follows from the argument above that $\procB$ cannot be stuck in one visit to its doorway,
  so it must be a repeating series of visits. This is impossible because when \procA tries to
read $\it choosing[\procB]$ for the first time, the value of $\it number[\procA]$ is set
and will not change
anymore, so if \procB goes back to the doorway, it is bound
to set $\it number[\procB] > number[\procA]$ and will not be able to
enter the critical section anymore.

Suppose that Process \procA is stuck at $\ell_{7}$. Any process \procB that enters the doorway will
receive a $\it number[\procB]$ strictly larger than
$number[\procA]$ and be stuck in the bakery. So if
\procA is stuck, eventually all processes are stuck at $\ell_{7}$, which is impossible since
every finite lexicographically ordered set has a minimal element.

\section{Conclusion, Related Work and Outlook}

This paper presents a minimal extension
of CCS in which Peterson's mutual exclusion protocol
can be modelled correctly, using a justness assumption only.
The signalling operator allows processes to emit signals that can be received by other processes.
The signalling process is not blocked by the emission of the signal, which means that
its actions are in no way postponed or affected by other processes reading the signal.
This property
is crucial to correctly model mutual exclusion.

Our process algebra, \emph{CCS with signals}, is strongly inspired by, and can be regarded as a
simplification of, Bergstra's \emph{ACP with signals} \cite{Bergstra88}. The idea of a signal as a
predicate on states, rather than a transition between states, stems from that paper.
However, the non-blocking nature of signals was not explored by Bergstra,  who writes
``The relevance of signals is not so much that process algebra without signals lacks expressive power''.
This point is disputed in the current paper.

CCS with signals is not the first process algebra with explicitly non-blocking read actions.
In \cite{CDV09} Corradini, Di Berardini \& Vogler add a similar operator to PAFAS~\cite{CVJ02}, a process algebra for modelling timed concurrent systems.
The semantics of this extension is justified in \cite{CDV11}. 
They show~\cite{CDV09} that this enables the
liveness property of Dekker's mutual exclusion algorithm~\cite{EWD35,EWD123}, modelled in PAFAS, when assuming
\emph{fairness of actions}, and in \cite{EPTCS54.4} they establish the same for Peterson's
  algorithm, while showing that earlier mutual exclusion algorithms by
  Dijkstra~\cite{Dijkstra65} and Knuth~\cite{Knuth66} lack the liveness property under fairness of actions.
Fairness of actions is
similar to our notion of justness---although formalised in a quite different way---except that all
actions are treated as being non-blocking. The notion of time plays an important role in the
formalisation of the results in \cite{CDV09,EPTCS54.4}, even if it is not used quantitatively. Our process
algebra can be regarded as a conceptual simplification of this approach, as it completely abstracts
from the concept of time, and hence is closer to traditional process algebras like CCS and CSP\@.

The accuracy of our extension depends highly on which memory model is considered as realistic.
It is well known that in weak or relaxed memory models, mutual exclusion protocols like Peterson's or
the bakery algorithm do not behave correctly; when employing a weak memory model, mutual exclusion
is handled on the hardware layer only---this is not covered here. An extremely plausible memory
model allows parallel non-blocking writing, but admits any value being
written when two parallel write actions overlap. This memory model is compatible with the bakery
algorithm, and with Peterson's algorithm for two processes, but---as we show---not for Peterson's
algorithm with $N{\geq}3$ processes. Instead one needs a form of sequential consistency, assuming
that parallel write actions, or a parallel read/write, behave as if they are executed in either order.

When postulating sequential consistency, it is plausible to assume some kind of mutual exclusion
between write actions being implemented in hardware. This in turn allows the possibility of a write
action being delayed in perpetuity because other processes are writing to the same variable.
Similarly, read actions could be blocked by a consistent flow of write actions.
A third type of blocking is that write actions can be obstructed by read actions. However, this
kind of blocking is questionable;
it could be that during a parallel read/write the write action wins, and only the read action gets postponed.

When assuming all three kinds of blocking, the CCS rendering of mutual exclusion
protocols---illustrated in \autoref{sec:Peterson}---is fully accurate, and by \cite{GH15} we
conclude that no such protocol can have the intended liveness property. When disallowing write
actions being blocked by read actions, but allowing write/write blocking, we get the modelling in
\CCSs. Using \CCSs, we verified the correctness of the bakery algorithm, and Peterson's
algorithm for two processes, whereas Peterson's for $N>2$ fails liveness.  The
latter protocol becomes correct if we assume sequential consistency without any kind of blocking.
Whether this is a realistic memory model on modern hardware needs further investigation.
Regardless, we conjecture that such a memory can be modelled in an extension of \CCSs with broadcast
communication, i.e.\ the combination of the process algebras presented here and in \cite{GH14}.

The liveness property of Dekker's algorithm, when assuming merely justness, or fairness of
actions, requires not only non-blocking reading, but also that repeated assignments to a variable $x$ of
the same~value cannot block the reading of $x$ \cite{CDV09}.
This assumption can be modelled in \CCSs,\vspace{-2pt} by defining $\rA$ of
\hyperlink{Ex1}{Ex.~1} by
$x^\tr \stackrel{\it def}{=} (\ass{x}{\fa}\mathbin.x^\fa)\signals\noti{x}{\tr}$ and
$x^\fa \stackrel{\it def}{=} (\ass{x}{\tr}\mathbin.x^\tr)\signals\noti{x}{\fa}$,
and replacing write actions $\overline{\ass{x}{v}}$ by $(\overline{\ass{x}{v}} + \noti{x}{v})$.
Alternatively, a pseudocode assignment $x:=v$ could be interpreted as ~$\textbf{if}~ x\neq v ~\textbf{then}~ x:=v ~\textbf{fi}$.

Although mutual exclusion protocols cannot be modelled in standard Petri nets---when not
assuming fairness---\cite{KW97,Vogler02,GH15}, it is possible in nets extended with read arcs
\cite{Vogler02}. This opens the possibility of interpreting \CCSs in terms of nets with read arcs,
whereas an accurate semantics of \CCSs in terms of standard nets is impossible.
A read arc from a place to a transition requires the place to be marked for the
transition be enabled, but the token is not consumed when the transition is fired.
This behaviour really looks like signalling, so
we conjecture that a read-arc net semantics of \CCSs is fairly straightforward.

Finally, the definition of justness appears
complicated because it includes the decomposition of paths. In order to compute if a path (an object
from the semantics) is just or not just, we
investigate the syntactic shape of the states on that path.
It could be that the semantic object---the labelled transition system%
---is not well adapted to the problem of justness. Giving a semantics to \CCSs
that inherently includes the decomposition of paths---inspired by \cite{BCHK93,BCHK94,Aceto94,MN92}---could be an interesting idea for future research.

\newpage
\bibliographystyle{eptcs}
\bibliography{aodv}
\end{document}